\documentclass[reqno]{amsart}
\usepackage{amsmath,amssymb}
\usepackage{graphicx}
\usepackage[colorlinks = true, linkcolor = blue, urlcolor  = blue, citecolor = red]{hyperref}
\usepackage{color}
\usepackage{multirow}

\usepackage[margin=1in]{geometry}

\newtheorem{theorem}{Theorem}[section]
\newtheorem{definition}[theorem]{Definition}
\newtheorem{proposition}[theorem]{Proposition}
\newtheorem{lemma}[theorem]{Lemma}

\newtheorem{remark}[theorem]{Remark}

\title{Thresholds for reduction-related entanglement criteria in quantum information theory}
\date{\today}
\author{Maria Anastasia Jivulescu}
\address{MAJ: Department of Mathematics,
Politehnica University of Timi\c soara,
Victoriei Square 2, 300006 Timi\c soara, Romania}
\email{maria.jivulescu@upt.ro}

\author{Nicolae Lupa}
\address{NL: Department of Mathematics,
Politehnica University of Timi\c soara,
Victoriei Square 2, 300006 Timi\c soara, Romania}
\email{nicolae.lupa@upt.ro}

\author{Ion Nechita}
\address{IN: Zentrum Mathematik, M5, Technische Universit\"at M\"unchen, Boltzmannstrasse 3, 85748 Garching, Germany
 and CNRS, Laboratoire de Physique Th\'{e}orique, IRSAMC, Universit\'{e} de Toulouse, UPS, F-31062 Toulouse, France}
\email{nechita@irsamc.ups-tlse.fr}

\subjclass[2010]{15B48, 81P45}
\keywords{quantum entanglement, reduction criterion, thresholds}

\date{\today}

\begin{document}

\begin{abstract}
We consider random bipartite quantum states obtained by tracing out one subsystem from a random, uniformly distributed, tripartite pure quantum state. We compute thresholds for the dimension of the system being traced out, so that the resulting bipartite quantum state satisfies the reduction criterion in different asymptotic regimes. We consider as well the basis-independent version of the reduction criterion (the absolute reduction criterion), computing thresholds for the corresponding eigenvalue sets. We do the same for other sets relevant in the study of absolute separability, using techniques from random matrix theory. Finally, we gather and compare the known values for the thresholds corresponding to different entanglement criteria, and conclude with a list of open questions.
\end{abstract}

\maketitle

\tableofcontents

\section{Introduction}

The notion of quantum entanglement  has been proved to be at the core of many quantum phenomena, such as teleportation, dense coding or cryptography. Moreover, it is a key ingredient to the computational power of quantum devices. Entanglement expresses inseparability, that is unusual correlations between the subsystems of a quantum system which cannot be explained by classical, non-quantum, models. Hence, a central question in the theory of quantum information and computation is detecting the presence and measuring the amount of entanglement present in a given quantum system.

One of the most efficient tools in detecting entanglement is the positive partial transpose $(\mathrm{PPT})$ criterion \cite{peres}.
It states that if a quantum state is separable, then the partial transpose with respect to one of the subsystems is positive-semidefinite.
It represents a necessary condition for separability and more often it is applied as a tool to detect entanglement:
if the partial transpose of a given state is not positive-semidefinite, then the state is entangled.
As the partial transposition criterion is obtained by applying the transposition operator over the second subsystem, it raised the question of finding other positive maps $P$ with the property $(\mathrm{id} \otimes P)(\rho)$ is not positive-semidefinite, for some entangled bipartite states $\rho$. One needs to look for such maps $P$ in the class of positive, but not completely positive applications. In \cite{horodeckiPPTcriterion}, the authors show the following converse:
a quantum state $\rho\in \mathcal M_n(\mathbb C)\otimes \mathcal M_k(\mathbb C)$ is separable if and only if $(\mathrm{id}\otimes {P})(\rho)$ is positive-semidefinite for \emph{all} positive maps ${P}:\mathcal M_k(\mathbb C)\to \mathcal M_m(\mathbb C)$ and all positive integers $m \in \mathbb N$ (actually, $m=n$ suffices, see \cite[Theorem 2]{horodeckiPPTcriterion}), where $\mathrm{id}$ is the identity map. Thus, each fixed positive map yields a necessary condition for separability; in other words, each fixed positive map yields an entanglement criterion. Note that these conditions (resp.~ criteria) become trivial for completely positive maps.

A possible choice of the positive map ${P}$ is the \emph{reduction map}
${R}:\mathcal M_k(\mathbb C)\to \mathcal M_k(\mathbb C)$, ${R}(X):=I_k\cdot\mathrm{Tr}(X)-X,$
and the corresponding separability test is called  \emph{reduction $(\mathrm{RED})$ criterion}  \cite{cerf,horodeckireduction}.
The reduction criterion is weaker than the $\mathrm{PPT}$ criterion: if a state violates the reduction criterion, then it also violates the $\mathrm{PPT}$ criterion \cite{horodeckireduction}. Conversely, there exist states (some entangled Werner states \cite{wernerentangled}) which satisfy the reduction criterion but violate the $\mathrm{PPT}$ criterion.
The two criteria are equivalent if the subsystem on which the reduction map is applied is a qubit \cite{cerf}.

Both PPT and reduction criteria are efficient theoretical tools in detecting quantum entanglement, although the former presents the obvious advantage of a more elegant form, which in addition requires less computations.
In practice however, since impure entanglement is produced,  the concept of distillation was introduced as the process to produce a pure maximally entangled state by local quantum operations and classical communication, from many copies of an arbitrary entangled state (see \cite{ben},  \cite[pp. 870]{horodeckireview}).
Horodecki proved that a PPT state is necessarily undistillable \cite{horo98}.
This result sheds light on the fact that in high dimensions there are entangled states which cannot be distilled.
These states, namely PPT entangled states, are called \emph{bound entangled states}, contrary to entangled states which can be distilled.
It is possible to show that reduction criterion and entanglement distillation are connected:  any state which violates the reduction criterion is distillable;  conversely, if a state can be distilled by a certain protocol, then the state violates the reduction criterion \cite{horodeckireduction}. This result justifies the use of the reduction criterion, even if, from a purely entanglement-detection perspective, it is weaker than the PPT criterion.

The separability problem was also approached by studying the class of  \emph{absolutely separable states} (ASEP), i.e.\  states  that
remain separable under any global unitary transformation  \cite{kus}, that means to find conditions on the spectrum that characterize absolutely separable states (constraints on the eigenvalues of a state $\rho$ guaranteeing that $\rho$ is separable with respect to any decomposition of the corresponding product tensor space \cite{knil}). This problem was first fully solved in the qubit-qubit case in \cite{verstraeteaudenaert}, and then in the qubit-qudit case in \cite{joh}.
Furthermore, it is known that there is an Euclidean ball of known radius centered at the maximally-mixed state $\frac{1}{nk} (I_n\otimes I_k)$ such that every state within this ball is separable \cite{gba} (see also \cite{zycz}), meaning that any state within this ball is actually absolutely separable. However, there exist absolutely separable states outside of this ball \cite[Appendix B]{vidaltarrach}.
In analogy to absolutely separable states, states which remain $\mathrm{PPT}/\mathrm{RED}$ under any global unitary transformation are called  \emph{absolutely $\mathrm{PPT}$ states} ($\mathrm{APPT}$) /\emph{absolutely $\mathrm{RED}$ states} ($\mathrm{ARED}$) \cite{zycz}.
Necessary and sufficient conditions on the spectrum of APPT-states are given in \cite{hil}, in the form of a finite set (albeit exponentially large in the dimension) of linear matrix inequalities.  For the case of ARED-states, necessary and sufficient conditions are given in the form of a infinite family of linear inequalities, which the spectrum has to verify \cite{jlnr}.

In this paper we approach the problem of separability and absolute-separability from a different perspective. We aim to derive \emph{thresholds} for the reduction and absolute reduction criteria and to give a complete picture of threshold points for the class of entanglement criteria.
The threshold point is defined in the following sense: given a random mixed state
$\rho_{AB} \in \mathcal M_n(\mathbb C) \otimes \mathcal M_k(\mathbb C)$, obtained by partial tracing over $\mathbb{C}^s$ a uniformly distributed, pure quantum state $x \in \mathbb C^n \otimes \mathbb C^k \otimes \mathbb C^s$, where the $s$-dimensional space is treated like an inaccessible environment,
we ask for the probability that the state satisfies an entanglement criterion.
When one (or both) of the system dimensions $n$ and $k$ are large, a \emph{threshold phenomenon} occurs: if $s \sim c\cdot f(n,k)$, for some constant $c>0$,  or $s$ is fixed, then there is a \emph{threshold value} $c_{0}$ of the scaling parameter, such that the following holds:
\begin{enumerate}
\item for all $c < c_{0}$, as dimension $nk$ grows, the probability that $\rho_{AB}$ satisfies the entangled criterion vanishes;
\item for all $c > c_{0}$, as dimension $nk$ grows, the probability that $\rho_{AB}$ satisfies the entangled criterion converges to one.
\end{enumerate}
The threshold phenomenon was introduced by Aubrun to study the PPT criterion \cite{aub}.
Our main contribution presented in this paper is to complete the computation of the thresholds for the reduction criterion given in \cite{jln} and to derive  the threshold for the absolute reduction criterion, in different asymptotic regimes.

The paper is organized as follows: Sections \ref{sec:criteria} and \ref{sec:random} aim to introduce the main concepts and notations used in the paper. In Sections \ref{sec:RED} and \ref{sec:ARED} we compute explicitly the value of thresholds for reduction and absolutely reduction criteria, in different asymptotic regimes.
In Sections \ref{sec:LSp}, \ref{sec:GER}, and \ref{sec:SEPBALL} we derive thresholds for some sets that express certain conditions on probability vectors and which approximate the set of separable states. In the last section, we gather all the results about thresholds for different sets, and we present open questions related to this subject.

\smallskip
{\noindent\it Acknowledgments.} The work of MAJ and NL was supported by a grant of the Romanian National Authority for Scientific Research, CNCS-UEFISCDI, project number PN-II-ID-JRP-2011-2/11-RO-FR/01.03.2013.
IN's research has been supported by a von Humboldt fellowship and by the ANR projects {OSQPI} {2011 BS01 008 01}, {RMTQIT}  {ANR-12-IS01-0001-01}, and {STOQ}  {ANR-14-CE25-0003}.

\section{Entanglement criteria}
\label{sec:criteria}
In this paper the set of density operators (positive-semidefinite matrices of unit trace) acting on $\mathbb C^d$  is denoted by $D_d$ and for
 bipartite quantum systems on tensor product Hilbert space $\mathbb C^n\otimes \mathbb C^k\cong\mathbb C^{nk}$  we identify $D_d$ with $D_{n,k}$,  subscripts indicating the bipartition ($n$ will denote the Hilbert space dimension of the first tensor factor and $k$ that of the second one, and both  $n,k\geq2$).

A density operator  $\rho\in \mathcal M_n(\mathbb C)\otimes \mathcal M_k(\mathbb C)$  (here  $\mathcal M_n(\mathbb C)$ denotes the space of all $n\times n$ complex matrices) is called \emph{separable} \cite{wernerentangled} if it can be written as
$$\rho=\sum_i p_i e_ie_i^* \otimes f_if_i^*$$
for $p_i\geq 0$, $\sum_i p_i=1$, and for unit vectors $e_i\in\mathbb{C}^n$, $f_i\in\mathbb{C}^k$ (throughout the paper we will identify quantum states with their density matrices).
The set of separable states \cite{wernerentangled,horodeckireview} in $D_{n,k}$  is denoted by
$$\mathrm{SEP}_{n,k}:=\{\rho\in D_{n,k}\,|\,\rho~\text{separable}\}.$$
Efficient methods for  explicit characterizations of $\mathrm{SEP}_{n,k}$ are not known and for this reason upper and lower approximations are of interest \cite{horodeckireview}.

On any matrix algebra $\mathcal M_d(\mathbb C)$, we define the \emph{reduction map},
\begin{align*}
R:\mathcal M_d(\mathbb C)\to \mathcal M_d(\mathbb C), \qquad R(X):=I_d\cdot\mathrm{Tr}(X)-X,
\end{align*}
where $I_d$ denotes an identity matrix of the appropriate dimension (here, $d$) and $\mathrm{Tr}$ is the usual, unnormalized, matrix trace. From the definition, it follows that the map $R$ is positive, i.e. $R(X) \geq 0$ whenever $X \geq 0$.

For a bipartite matrix $X= X_{AB} \in \mathcal M_n(\mathbb C) \otimes \mathcal M_k(\mathbb C)\cong \mathcal M_{nk}(\mathbb C)$, its reduction over the second subsystem ($B$) is denoted by
$$X^{red} := ({\rm id}\otimes R)(X_{AB}) = X_A \otimes I_k - X_{AB},$$
where $X_A: = (\mathrm{id} \otimes \mathrm{Tr})(X)$ denotes the partial trace over $(B)$ of the operator $X=X_{AB}$. We write the transposition map on any matrix algebra $M_d(\mathbb C)$ as $\Theta$, and we also write $\Theta(X)= X^T$; we denote the \emph{partial transposition} of a bipartite matrix $X=X_{AB}$ by
$$X^\Gamma:=({\rm id}\otimes\Theta)(X).$$
The composition of $\Theta$ with the completely positive map $R\Theta:X\mapsto I_d\cdot\mathrm{Tr}(X)-\Theta(X)$ is the reduction map $R$ defined above; one says that the reduction map $R$ is completely co-positive.

Every positive map $P$ on $\mathcal M_k(\mathbb C)$ defines an entanglement criterion \cite{horodeckiPPTcriterion,horodeckireview}: if, for $\rho\in D_{n,k}$, the matrix $({\rm id}\otimes P)(\rho)$ is not positive-semidefinite, then $\rho$ is entangled. Specializing to the reduction map $P=R$, this becomes the \emph{reduction criterion} \cite{horodeckireduction,cerf}, which is also related to the distillability of the state in question \cite{horodeckireview}. Every bipartite state whose entanglement is detected by the reduction criterion is also detected by the partial transposition criterion \cite{peres,horodeckiPPTcriterion}, which is the above criterion for the map $P=\Theta$; this follows from the above mentioned representation of $R$ as the composition of $\Theta$ with a completely positive map.

The set of density operators $\rho\in D_{n,k}$ having positive reductions with respect to the second tensor factor for the fixed tensor decomposition $\mathcal M_{nk}(\mathbb C) \cong \mathcal M_n(\mathbb C) \otimes \mathcal M_k(\mathbb C)$ is denoted by
$$\mathrm{RED}_{n,k}:= \{\rho \in D_{n,k} \,|\, \rho^{red} \geq 0\}.$$
The entanglement criterion based on positive maps \cite{horodeckiPPTcriterion} implies the inclusion $\mathrm{SEP}_{n,k}\subseteq\mathrm{RED}_{n,k}$ \cite{horodeckireduction,cerf}. Recall also that the set of states with positive partial transpose is
$$\mathrm{PPT}_{n,k}:=\{\rho\in D_{n,k}\,|\,\rho^\Gamma\geq0\}.$$
Note that, when $k=2$, the reduction and the PPT criterion are equivalent \cite{horodeckireduction,cerf,jln}, i.e.\ they detect entanglement for the same states, so that $\mathrm{PPT}_{n,2}=\mathrm{RED}_{n,2}$; in general, the reduction criterion is weaker: $\mathrm{SEP}_{n,k} \subseteq \mathrm{PPT}_{n,k} \subseteq \mathrm{RED}_{n,k}$. Furthermore, it is well known that $\mathrm{SEP}_{n,k}=\mathrm{PPT}_{n,k}$ whenever $nk\leq6$ \cite{horodeckiPPTcriterion}. Occasionally we will write $\mathrm{RED}$ instead of $\mathrm{RED}_{n,k}$ etc., as the dimensions of the subsystems will be clear from the context most of the time. For a sketch of the different sets corresponding to the criteria described above and their inclusions, see Figure \ref{fig:sets}; the figure on the left contains the set $\mathrm{RLN}$ of states satisfying the \emph{realignment} criterion \cite{cwu} (also known as the computable cross-norm criterion \cite{rud}).

\begin{figure}[htbp]
\begin{center}
\includegraphics[width=0.45\textwidth]{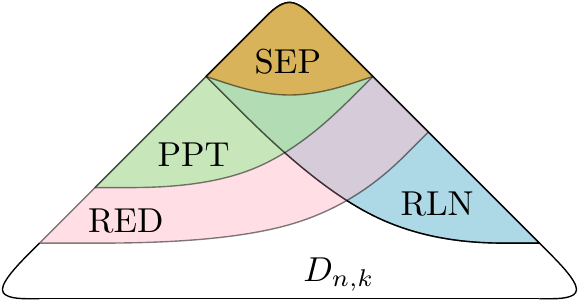} \quad \includegraphics[width=0.45\textwidth]{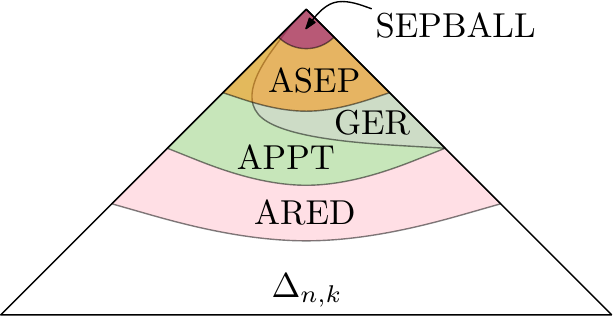}
\caption{On the left, subsets of $D_{n,k}$ corresponding to entanglement criteria and their inclusions. On the right, subsets of $\Delta_{n,k}$ corresponding to the ``absolute'' versions of the entanglement criteria, as well as the sets $\mathrm{GER}$ and $\mathrm{SEPBALL}$, see Sections \ref{sec:GER} and \ref{sec:SEPBALL}.}
\label{fig:sets}
\end{center}
\end{figure}

Let us now introduce the ``absolute'' versions of the entanglement criteria above. We denote by $\mathcal U_{nk}$ the set of unitary operators acting on $\mathbb C^{nk}$. The set of states which remain $\mathrm{RED}$ under any global unitary transformation $U\in\mathcal U_{nk}$ is denoted by $\mathrm{ARED}$ (``\emph{absolutely} $\mathrm{RED}$''):
$$\mathrm{ARED}_{n,k}~:=~\{ \rho \in D_{n,k} \,|\, \forall U \in \mathcal U_{nk}: \, (U\rho U^*)^{red} \geq 0 \} = \bigcap_{U \in \mathcal U_{nk}} U \mathrm{RED}_{n,k} U^*.$$
Similarly:
\begin{align*}
\mathrm{APPT}_{n,k}~&:=~\bigcap_{U\in\mathcal U_{nk}}U\mathrm{PPT}_{n,k}U^*\,,\\
\mathrm{ASEP}_{n,k}~&:=~\bigcap_{U\in\mathcal U_{nk}}U\mathrm{SEP}_{n,k}U^*\,.
\end{align*}
Obviously, $\mathrm{ASEP}_{n,k}\subseteq \mathrm{APPT}_{n,k}\subseteq \mathrm{ARED}_{n,k}$ and $\mathrm{AX}_{n,k}\subseteq X_{n,k}$, for $\mathrm{X}=\mathrm{SEP, PPT}$ and $\mathrm{RED}$. The question whether a quantum state $\rho$ belongs in one of the three sets introduced above depends only on the spectrum of $\rho$; this is why sometimes we identify the sets $\mathrm{AX}$ with sets of spectra:
$$\mathrm{AX}_{n,k} \subseteq \Delta_{nk}:=\{ \lambda \in \mathbb R^{nk}_+ \, | \, \sum_{i=1}^{nk} \lambda_i = 1\}.$$
There are known results on the characterizations of the sets $\mathrm{APPT}_{n,k}$ and  $\mathrm{ARED}_{n,k}$ as given by necessary and sufficient conditions in the form of families of linear inequalities which the spectrum has to verify \cite{hil, jlnr}. To date, there is no simple characterization of $\mathrm{ASEP}$; in \cite{joh}, Johnston shows that $\mathrm{ASEP}_{n,2}=\mathrm{APPT}_{n,2}$ (and thus, also equal to $\mathrm{ARED}_{n,2}$, see \cite[Proposition 5.1]{jlnr}), while in \cite{ajr} further evidence towards the conjecture that $\mathrm{ASEP}_{n,k}=\mathrm{APPT}_{n,k}$ is presented.

\section{On the spectrum of large Wishart matrices and random induced states}
\label{sec:random}
In the following, if $(a_n)$ and $(b_n)$ are some nonzero sequences, as usual, $a_n\sim b_n$ means that $a_n/b_n\to 1$ as $n\to\infty$, while $a_n=o(b_n)$ means that $a_n/b_n\to 0$ as $n\to\infty$; also, $a_n\ll b_n$ means that $b_n/a_n\to \infty$ as $n\to\infty$.

Let us first recall the notion of Wishart ensemble of random matrices:

\begin{definition}\label{def.Wishart}\rm A \emph{Wishart  matrix} of parameters $(d,s)$ is a random $d\times d$ matrix $W$ given by $W=GG^*$, where
$G\in\mathcal M_{d\times s}(\mathbb C)$  is a $d\times s$   matrix whose entries are i.i.d.~complex Gaussian random variables of zero mean and unit variance.
\end{definition}

The following well-known result describes the behavior of the spectrum of a Wishart matrix of parameters $(d,s)$ in the asymptotic regime $d\to\infty$ and $s/d\to c\in(0,\infty)$:

\begin{proposition} \label{prop:extremal-Wishart-MP}
Let $\lambda_1\geq \lambda_2\geq \cdots \geq \lambda_d\geq 0$ be the eigenvalues of a Wishart matrix of parameters $(d,s)$. Then, in the asymptotic regime $d\to \infty$ and $s=s_d\sim cd$ for some constant $c>0$, one has:
\begin{enumerate}
\item Almost surely, as $d\to\infty$, the renormalized empirical eigenvalue distribution
$$ \mu_d=\frac{1}{d}\sum_{i=1}^{d}\delta_{d^{-1}\lambda_i}$$
converges weakly to the Mar\v{c}enko-Pastur distribution
$$\pi_c=\max (1-c,0)\delta_0+\sqrt{4c-(x-1-c)^2} \,\mathbf{1}_{[(\sqrt{c}-1)^2,(\sqrt{c}+1)^2]}(x) dx; $$
\item  For any function $j_d = o(d)$, almost surely, as $d\to\infty$, the rescaled eigenvalues $\tilde \lambda_i = d^{-1}\lambda_i$ have the following limits
\[
\tilde \lambda_d,\tilde \lambda_{d-1},\ldots,\tilde \lambda_{d-j_d+1}\to a_c=
\begin{cases}
0,& \text{if } c\leq 1,\\
(\sqrt{c}-1)^2, &\text{if } c>1,
\end{cases}
\]
and
\[
\tilde \lambda_1,\tilde \lambda_{2},\ldots,\tilde \lambda_{j_d}\to b_c=(\sqrt{c}+1)^2;
\]
\item  For every \emph{fixed} fraction $p \in (0,1)$, almost surely, as $d\to\infty$, the rescaled eigenvalue $\tilde \lambda_{\lfloor pd \rfloor}$ converges to the $(1-p)$-th quantile of the Mar\v{c}enko-Pastur distribution:
$$\lim_{d \to \infty} \tilde \lambda_{\lfloor pd \rfloor} = q_{1-p},$$
where $q_{1-p}$ is uniquely defined by
$$\int_{a_c}^{q_{1-p}} d \pi_c = 1-p.$$
\end{enumerate}
\end{proposition}
\begin{proof}
The convergence in distribution stated in the first point is the classical result of Mar\v{c}enko and Pastur \cite{mpa}. The convergence of the extreme eigenvalues of Wishart matrices has been shown by Bai and Yin  \cite[Theorem 2]{byi}, see also \cite[Theorem 5.11]{bsi}. The convergence of the eigenvalues in the bulk toward the corresponding quantile follows from the continuity of the distribution function of the limiting probability measure $\pi_c$ (except at the eventual atom in 0), see also \cite[Problem 2.4.19]{psh}. This also implies the convergence of the top (resp.~ bottom) $o(d)$ eigenvalues towards the edges of the support of $\pi_c$.
\end{proof}

We consider next the asymptotic regimes where $d \ll s$ and, respectively, $s \ll d$.

\begin{proposition}\label{prop:Wishart-dirac}
Consider a sequence of random matrices $(W_s)$, where $W_s \in \mathcal M_{d_s}(\mathbb C)$ is a Wishart matrix of parameters $(d_s, s)$, with $d_s=o(s)$. Then,
\begin{equation}\label{eq:Wishart-large-s}
\forall \varepsilon >0, \quad \lim_{s \to \infty} \mathbb P \left[ \left\| \frac{W_s}{s} - I_{d_s} \right\| \leq \varepsilon \right] = 1.
\end{equation}
Similarly, if $(\widetilde W_d)$ is a sequence of Wishart matrices of parameters $(d, s_d)$, such that $s_d = o(d)$, then
\begin{equation}\label{eq:Wishart-large-d}
\forall \varepsilon >0, \quad \lim_{d \to \infty} \mathbb P \left[ \left | \frac{\|\widetilde{W_d}\|}{d} - 1\right|  \leq \varepsilon \right] = 1.
\end{equation}
\end{proposition}

The first statement in the result above shows that in the regime $ d \ll s$, the eigenvalues of $W_s/s$ converge to $1$ as $s\to\infty$, in probability (see \cite{nec} and the references therein). We recall the following result from \cite{cny}, which shows that the fluctuations of these eigenvalues around $1$ are semicircular.

\begin{proposition}\label{prop:Wishart-centered-normalized}
 Let $W_{d}$ be a
Wishart matrix of parameters $(d,s)$ with $s=s_d$, and let
$$Z_{d}=\sqrt{ds}\left(\frac{W_{d}}{ds}-\frac{I_d}{d}\right)$$
be its centered and renormalized version. In the asymptotic regime $1 \ll d \ll s$ (i.e.~ $s/d \to \infty$ as $d \to \infty$), the random matrix $Z_d$ converges, in moments, to a standard semicircular distribution. Moreover, for any function $j_d = o(d)$, the top (resp.~ bottom) $j_d$ eigenvalues of $Z_d$ converge to $2$ (resp.~ $-2$): almost surely, as $d \to \infty$,
\begin{align*}
\lambda^\downarrow_1 (Z_d), \ldots, \lambda^\downarrow_{j_d}(Z_d) &\to 2 \\
\lambda^\downarrow_d (Z_d), \ldots, \lambda^\downarrow_{d-j_d+1}(Z_d) &\to -2.
\end{align*}
\end{proposition}
\begin{proof}
The convergence in moments has been shown in \cite[Corollary 2.5]{cny}, while the convergence of the operator norm of $Z_d$ to $2$ has been shown in \cite[Theorem 2.7]{cny}. The extension of the norm convergence to that of a $o(d)$ number of eigenvalues is a classical argument in random matrix theory, see also the proof of Proposition \ref{prop:extremal-Wishart-MP} and \cite[Problem 2.4.19]{psh}.
\end{proof}

\bigskip
We consider now the standard model of \emph{random induced quantum states}.
Let $\psi$ be a random unit vector uniformly distributed on the unit sphere in $\mathbb{C}^d\otimes \mathbb{C}^s$.
We denote by $\mu_{d,s}$ the distribution of the random quantum state
$$\rho=\mathrm{Tr}_{\mathbb{C}^s}\left(\psi\psi^*\right) \in D_d,$$
obtained after partial tracing over the ancilla space $\mathbb{C}^s$ the Haar-distributed random pure state $\psi\psi^*$, which is called the \emph{induced measure} of parameters $(d,s)$. Any random state $\rho$ distributed according to the induced measure $\mu_{d,s}$ is called a \emph{random induced state} and we say that $\rho$ comes from the \emph{induced ensemble} with parameters $(d,s)$.

Random induced states are closely related to Wishart matrices. Indeed, it follows from \cite[Lemma 1]{nec} that if $W$ is a Wishart matrix of parameters $(d,s)$, then $\rho:=\frac{W}{\mathrm{Tr} W}$ is a random state with distribution $\mu_{d,s}$. Therefore, results about Wishart matrices can be translated to random induced states, as we can see from the following result (for more details, we refer the reader to \cite{aub} or \cite{nec}):
\begin{proposition}\label{prop.trace}
If $W$ is a Wishart matrix of parameters $(d,s)$, then for every $\varepsilon>0$,
$$\mathbb{P} \left[\left|\frac{\mathrm{Tr}W}{ds}-1\right|>\varepsilon\right]\leq C \exp(-c ds \varepsilon^2),$$
for some  $C,c>0$.
\end{proposition}
The main advantage of this approach is that the distribution of a Wishart matrix is much easier to deal with than the induced measure $\mu_{d,s}$.

\bigskip

In the rest of the paper, we shall be interested in random bipartite induced quantum states: we shall assume that $d = nk$, such that we have a framework for studying the entanglement of the random states $\rho \sim \mu_{nk,s}$. Our techniques come from random matrix theory and are adapted to the study of large dimensional states ($nk \to \infty$). We shall thus consider three asymptotical regimes:
\begin{enumerate}
\item The \emph{balanced} regime: both $n,k \to \infty$;
\item The \emph{first unbalanced} regime: $n$ is fixed and $k \to \infty$;
\item The \emph{second unbalanced} regime: $k$ is fixed and $n \to \infty$.
\end{enumerate}

\section{\texorpdfstring{Thresholds for $\mathrm{RED}$ in the balanced and the first unbalanced case}{Thresholds for RED in the balanced and the first unbalanced case}}
\label{sec:RED}

We discuss in this section some questions left open in \cite{jln} regarding the thresholds for the set $\mathrm{RED}$, in two asymptotic regimes. Indeed, in \cite[Sections VII, VIII]{jln} it has been shown that for any size of the environment $s$ which behaves like $s \sim cnk$ for some positive constant $c >0$, with overwhelming probability, quantum states $\rho$ distributed along the induced measure with parameters $(nk,s)$, will satisfy the reduction criterion, in the regimes where $n$ grows and $k$ is either fixed or it grows at the same speed as $n$ (see \cite[Theorem 7.2 and Theorem 8.2]{jln}).
Hence, the thresholds for the set $\mathrm{RED}$ in these asymptotic regimes must be of smaller order than $nk$. We compute next the exact regimes and threshold values for these cases.

\begin{theorem}\label{thm:RED-balanced}
Consider a sequence of random density matrices $(\rho_n)$,  where $\rho_n$ comes from the induced ensemble with parameters $(nk_n, s_n)$. In the \emph{balanced regime}, where $n\to \infty$, $k_n \to \infty$ as $n\to \infty$ (not necessarily at the same speed) and $s_n \sim cn$ for some constant $c>0$, we have (below, $\mathbb P=\mathbb P_{n}$, as a function of $n$, denotes the probability distribution of $\rho_n$):
\begin{enumerate}
\item If $c<1$, then $\lim_{n \to \infty} \mathbb P[\rho_n \text{ has positive reduction}] = 0$;
\item If $c>1$, then $\lim_{n \to \infty} \mathbb P[\rho_n \text{ has positive reduction}] = 1$.
\end{enumerate}
In other words, the threshold for the reduction criterion in the balanced regime is $c=1$, on the scale $s_n \sim cn$.
\end{theorem}
\begin{proof}
Instead of working with the induced measure for random quantum states, we shall use the simpler Wishart ensemble, since the reduction criterion is scale invariant:
$$(\mathrm{id} \otimes R)\left( \frac{W}{\mathrm{Tr} W} \right) \geq 0 \iff (\mathrm{id} \otimes R)(W) \geq 0.$$
To this end, consider a sequence $(W_n)$ of Wishart matrices of parameters $(nk_n, s_n)$, and define
$$Q_n = W_{A,n} \otimes I_{k_n} - W_n.$$
Let us first assume that $c>1$. Choose any $\varepsilon>0$ small enough such that $c(1-\varepsilon) > 1+\varepsilon$ and thus, for $n$ large enough, $s_n (1 - \varepsilon) > n(1+\varepsilon)$. From the partial trace property of Wishart matrices, it follows that the matrix $W_{A,n}$ follows a Wishart distribution of parameters $(n, k_ns_n)$. Since $n = o(k_n s_n)$, it follows by Proposition \ref{prop:Wishart-dirac} \eqref{eq:Wishart-large-s} that
\begin{equation}\label{eq:c-g-1-1}
\lim_{n \to \infty}  \mathbb P \left[ \frac{W_{A,n} \otimes I_{k_n} }{k_n s_n} \geq (1-\varepsilon)I_{nk_n} \right] = 1.
\end{equation}
Similarly, using Proposition \ref{prop:Wishart-dirac} \eqref{eq:Wishart-large-d}, we have
\begin{equation}\label{eq:c-g-1-2}
\lim_{n \to \infty}  \mathbb P \left[ \frac{W_n }{nk_n } \leq (1+\varepsilon)I_{nk_n} \right] = \lim_{n \to \infty}  \mathbb P \left[ \frac{W_n }{k_n s_n } \leq \frac{n(1+\varepsilon)}{s_n}I_{nk_n} \right]  =1.
\end{equation}
We conclude in the case $c>1$ by combining equations \eqref{eq:c-g-1-1} and \eqref{eq:c-g-1-2}.

Let us now move on to the case $c<1$, and show that, in this regime, with large probability as $n \to \infty$, the random matrix $Q_n$ is not positive-semidefinite. We proceed in a similar fashion: using Proposition \ref{prop:Wishart-dirac}, we have
\begin{equation}\label{eq:c-s-1-1}
\lim_{n \to \infty}  \mathbb P \left[ W_{A,n} \otimes I_{k_n}  \leq (1+\varepsilon) k_n s_n I_{nk_n} \right] = 1
\end{equation}
and
\begin{equation}\label{eq:c-s-1-2}
\lim_{n \to \infty}  \mathbb P \left[ \|W_n \| \geq (1-\varepsilon)nk_n \right] =1.
\end{equation}
Choosing $\varepsilon>0$ small enough such that, for $n$ large enough, $ (1-\varepsilon)n >  (1+\varepsilon)s_n$, and using \eqref{eq:c-s-1-1}-\eqref{eq:c-s-1-2}, we can conclude.
\end{proof}

\begin{theorem}\label{thm:RED-first-unbalanced}
Consider the \emph{first unbalanced regime}, where $n$ and $s$ are fixed integers, and $k \to \infty$. Let $(\rho_k)$ be a sequence of quantum states, where $\rho_k$ comes from the induced ensemble with parameters $(nk,s)$.
\begin{enumerate}
\item If $s<n$, then $\lim_{k \to \infty} \mathbb P[\rho_k \text{ has positive reduction}] = 0$.
\item If $s>n$, then $\lim_{k \to \infty} \mathbb P[\rho_k \text{ has positive reduction}] = 1$.
\end{enumerate}
In other words, the threshold for the reduction criterion in the first unbalanced regime is $s=n$, on the scale of bounded $s$. Here, $\mathbb P=\mathbb P_{k}$ is a function of $k$ and denotes the probability distribution of $\rho_k$.
\end{theorem}
\begin{proof}
The proof of this result is similar to that of Theorem \ref{thm:RED-balanced}. Working with Wishart matrices instead of random quantum states, the matrices $W_k$ and $W_{A,k}$ have Wishart distributions of respective parameters $(nk,s)$ and $(n,ks)$. Assume first that $s>n$. From Proposition \ref{prop:Wishart-dirac}, it follows that
\begin{equation}\label{eq:s-g-n-1}
\lim_{k \to \infty}  \mathbb P \left[ W_{A,k} \otimes I_{k}  \geq (1-\varepsilon) k s I_{nk} \right] = 1
\end{equation}
and
\begin{equation}\label{eq:s-g-n-2}
\lim_{k \to \infty}  \mathbb P \left[ \|W_k \| \leq (1+\varepsilon)nk \right] =1.
\end{equation}
Choosing $\varepsilon>0$ small enough such that $ (1-\varepsilon)s >  (1+\varepsilon)n$, and using \eqref{eq:s-g-n-1}-\eqref{eq:s-g-n-2}, we can conclude. We leave the details of the case $s<n$ to the reader.
\end{proof}

Note that the second unbalanced regime, where $k$ is fixed and $n \to \infty$, has been treated in \cite[Proposition 10.3]{jln}: the threshold, on the scale $s_n \sim cnk$, is given by
$$c=\frac{(1+\sqrt{k+1})^2}{k(k-1)}.$$

\section{\texorpdfstring{Thresholds for $\mathrm{ARED}$}{Thresholds for ARED}}
\label{sec:ARED}

Let us first recall some notations and results from \cite{jlnr}.
For a  vector $x \in \mathbb R^r_+$,  we denote by $\operatorname{rk} x$ the number of non-zero elements of $x$.
\begin{definition}[``\emph{Hat operation}'' $x\mapsto\hat x$]\label{def:hat-xi}\rm
Given  $n,k\geq 2$ and a vector $x \in \Delta_r$ with $r \leq \min(n,k)$, we associate to $x$ the pure quantum state $\psi \in \mathbb C^n \otimes \mathbb C^k$ given by
$$\psi = \sum_{i=1}^r \sqrt{x_i} e_i \otimes f_i,$$
where $(e_i)_{i=1}^n$ and $(f_j)_{j=1}^k$ are fixed orthonormal families in $\mathbb C^n$ and \ $\mathbb C^k$, respectively.
We then define $\hat x$ to be the vector of eigenvalues of the reduction $(\psi\psi^*)^{red}$ of the state $\psi\psi^*\in D_{n,k}$, taken with multiplicities as in Corollary 3.3 in \cite{jlnr}:
$$\hat x ~:=~ (\underbrace{x_1, \ldots, x_1}_{k-1 \text{ times}}, \eta_1 ,  \underbrace{x_2, \ldots, x_2}_{k-1 \text{ times}}, \ldots, \eta_{r-1}, \underbrace{x_r, \ldots, x_r}_{k-1 \text{ times}}, \underbrace{0, \ldots, 0}_{(n-r)k \text{ times}},\eta_r)\,\in\,\mathbb R^{nk},$$
where $x_i\geq\eta_i\geq x_{i+1}$ for $i\in[r-1]$ and $\eta_r=-\sum_{i=1}^{r-1}\eta_i\leq0$. The set $\{\eta_i\}_{i=1}^r\setminus\{x_i\}_{i=1}^r$ equals the set of solutions $\eta\in\mathbb R\setminus\{x_i\}_{i=1}^r$ to the equation $\sum_{i=1}^r \frac{x_i}{x_i-\eta}=1$ (for more details we refer the reader to Theorem 3.1 in \cite{jlnr}). Moreover, if $r=\operatorname{rk} x$, then $-\eta_r\leq (1-1/r)\sum_{i=1}^r x_i$, with equality if and only if $x_i=x_j$ for all $i,j\in[r]$ (see Lemma 3.2 in \cite{jlnr}).
\end{definition}

A characterization of the set $\mathrm{ARED}_{n,k}$ has been given in Theorem 4.2 in \cite{jlnr} and states as follows:
\begin{proposition}\label{prop:ared}
We have
\begin{equation}\label{eq:ared}
\mathrm{ARED}_{n,k}\,=\, \{\rho \in D_{n,k} \, | \, \forall x\in \Delta_{\min(n,k)}, \, \langle \lambda_\rho^\downarrow,\hat x^\uparrow \rangle \geq 0 \},
\end{equation}
where $\lambda_\rho^\downarrow$ is the vector of eigenvalues of $\rho$ ordered decreasingly and $\hat x^\uparrow$ is the increasingly ordered version of $\hat x$ that has been introduced in Definition \ref{def:hat-xi}.
\end{proposition}

Let us make two general remarks about the set $\mathrm{ARED}_{n,k}$, which provide upper and lower bounds for this set:
$$ \{\lambda \, | \, \lambda_1^\downarrow \leq (k+1)\lambda_{nk}^\downarrow\} \subseteq \mathrm{ARED}_{n,k} \subseteq  \{\lambda \, | \, \operatorname{rk} \lambda \geq (n-2)k+2\}.$$

\begin{lemma}\label{lem:minimal-rank-ARED}
For any probability vector $\lambda \in \mathrm{ARED}_{n,k}$, we have $\operatorname{rk} \lambda \geq (n-2)k+2$.
\end{lemma}
\begin{proof}
Although the statement follows from the inclusion $\mathrm{ARED}_{n,k} \subseteq \mathrm{LS}_{2k-1}$ in \cite[Theorem 8.1]{jlnr}, we give here a direct proof. We assume $n \geq 2$ in order to avoid degenerate situations. Consider the vector $x = (1/2,1/2,0, 0,\ldots, 0) \in \Delta_{\min(n,k)}$. Using the ``hat operation''  from Definition \ref{def:hat-xi}, we have
$$\hat x = (\underbrace{1/2, \ldots, 1/2}_{2k-1 \text{ times}}, \underbrace{0, \ldots, 0}_{(n-2)k \text{ times}}, -1/2) \in \mathbb R^{nk}.$$
If $\operatorname{rk} \lambda \leq (n-2)k+1$, then, obviously, $\langle \lambda^\downarrow, \hat x^\uparrow \rangle = -\lambda^\downarrow_1/2 <0$, and thus $\lambda \notin \mathrm{ARED}_{n,k}$.
\end{proof}

\begin{lemma}\label{lem:multiple-smallest-lambda}
For any probability vector $\lambda \in \Delta_{nk}$, such that $\lambda_1^\downarrow \leq (k+1) \lambda_{nk}^\downarrow$, we have $\lambda \in \mathrm{ARED}_{n,k}$.
\end{lemma}
\begin{proof}
For any $x \in \Delta_{\min(n,k)}$ of rank $r=\operatorname{rk} x$, we have
\begin{align*}
\langle \lambda^\downarrow,\hat x^\uparrow \rangle&=x_1\left(\lambda^\downarrow_{nk}+\cdots+\lambda^\downarrow_{(n-1)k+2}\right)+\eta_1\lambda^\downarrow_{(n-1)k+1}   \\
&\quad +x_2\left(\lambda^\downarrow_{(n-1)k}+\cdots+\lambda^\downarrow_{(n-2)k+2}\right)+\eta_2\lambda^\downarrow_{(n-2)k+1}   \\
&\quad \cdots \\
&\quad +x_r\left(\lambda^\downarrow_{(n-r+1)k}+\cdots+\lambda^\downarrow_{(n-r)k+2}\right)
+\eta_{r}\lambda^\downarrow_1 \\
&\geq (k-1)(x_1+ \cdots + x_r)\lambda_{nk}^\downarrow + (\eta_1 + \cdots + \eta_{r-1})\lambda_{nk}^\downarrow + \eta_{r}\lambda^\downarrow_1 \\
&= (k-1) \lambda_{nk}^\downarrow + \eta_r(\lambda_{1}^\downarrow - \lambda_{nk}^\downarrow)\\
&\geq (k-1) \lambda_{nk}^\downarrow - (1-1/r)(\lambda_{1}^\downarrow - \lambda_{nk}^\downarrow)\\
&\geq (k-1) \lambda_{nk}^\downarrow - (1-1/k)(\lambda_{1}^\downarrow - \lambda_{nk}^\downarrow)\\
&= \frac{k-1}{k} \left[(k+1)\lambda_{nk}^\downarrow -\lambda_{1}^\downarrow \right]\geq 0.
\end{align*}
By Proposition \ref{prop:ared} it follows that $\lambda\in \mathrm{ARED}_{n,k}$, which concludes the proof.
\end{proof}

\subsection{\texorpdfstring{Threshold for $\mathrm{ARED}$ in the balanced case}{Threshold for ARED in the balanced case}}

\begin{theorem}\label{thm:threshold-ARED-balanced}
Consider the \emph{balanced} asymptotic regime, where  $n\to \infty$ and $k_n\to\infty$ as $n\to\infty$, and write $s=s_n$ as a function of $n$. Let $\rho_n$ be a random induced state distributed according to the induced measure $\mu_{nk_n,s_n}$. Almost surely, as $n\to\infty$ and $s_n\sim c nk_n$ for some constant $c>0$,
\begin{enumerate}
\item  if $c>1$, then $\rho_n \in \mathrm{ARED}_{n,k_n}$;
\item  if $c<1$, then $\rho_n \not\in \mathrm{ARED}_{n,k_n}$.
\end{enumerate}
\end{theorem}
\begin{proof}
Let us start with the easier, second point: if $c<1$ then, for $n$ large enough, we have  $s_n < (n-2)k_n+2$, so, by Lemma \ref{lem:minimal-rank-ARED}, no eigenvalue vector $\lambda$ sampled from the induced measure with parameters $(nk_n, s_n)$ will be an element of $\mathrm{ARED}_{n,k_n}$ (with non-zero probability).

In the case $c>1$, we show that the hypothesis of Lemma \ref{lem:multiple-smallest-lambda} is satisfied, almost surely as $n \to \infty$. Indeed, using Proposition \ref{prop:extremal-Wishart-MP}, we have the following almost sure limit:
$$\lim_{n \to \infty} s_n \left( \lambda_{nk_n}^\downarrow - \frac{ \lambda_{1}^\downarrow}{k_n+1} \right) = a_c-0 > 0,$$
and thus, by Lemma \ref{lem:multiple-smallest-lambda}, it follows that $\rho_n \in \mathrm{ARED}_{n,k_n}$.
\end{proof}

The result above states that the threshold for $\mathrm{ARED}$ in the balanced regime is $c=1$, on the scale $s_n\sim c nk_n$.

\subsection{\texorpdfstring{Threshold for $\mathrm{ARED}$ in the first unbalanced case}{Threshold for ARED in the first unbalanced case}}

\begin{theorem}\label{thm:threshold-ARED-1st-unbalanced}
Consider the \emph{first unbalanced} asymptotic regime, where $n$ is fixed and $k \to \infty$. Let $\rho_k$ be a random induced state distributed according to the induced measure $\mu_{nk,s_k}$. Almost surely, as $k\to\infty$ and $s_k\sim ck$  for some constant $c>0$, one has:
\begin{enumerate}
\item If $c > n-2$, then $\rho_k \in \mathrm{ARED}_{n,k}$;
\item If $c < n-2$, then $\rho_k \not\in \mathrm{ARED}_{n,k}$.
\end{enumerate}
In other words, the threshold for $\mathrm{ARED}$ in the first unbalanced regime is $c=n-2$, on the scale $s_k \sim ck$.
\end{theorem}
\begin{proof}
Again, the second point follows from Lemma \ref{lem:minimal-rank-ARED}. Let us prove the first statement.
To this end, we need to check that for all vectors $x \in \Delta_n$, $\langle \lambda^\downarrow_{\rho_k}, \hat x^\uparrow \rangle \geq 0$. If $\operatorname{rk} x = 1$, the previous inequality is satisfied, since all the components of $\hat x$ are non-negative. Consider now a vector $x$ with $r:=\operatorname{rk} x \geq 2$, and let $\varepsilon \in (0,1)$ such that $c > n-2+\varepsilon$. Define $t_k := \min(s_k, (n-1)k)$. For $k$ large enough, $t_k\geq (n-2+\varepsilon)k$. Then, we have
\begin{align}
\nonumber \langle \lambda^\downarrow_{\rho_k}, \hat x^\uparrow \rangle &\geq \lambda^\downarrow_1 \eta_r + x_2(\lambda^\downarrow_{(n-2)k+2} + \nonumber \lambda^\downarrow_{(n-2)k+3} + \cdots + \lambda^\downarrow_{t_k}) \\
\nonumber & \geq  \lambda^\downarrow_1 \eta_r + x_2 (\varepsilon k -1)\lambda_{t_k}^\downarrow  \\
\nonumber &\geq  \lambda^\downarrow_1 \eta_r + x_2 (\varepsilon k -1)\lambda_{s_k}^\downarrow  \\
\label{eq:bound-ARED-1st-unbalanced} &\geq  -\lambda^\downarrow_1  + n^{-1} (\varepsilon k -1)\lambda_{s_k}^\downarrow  .
\end{align}
Note that the quantity appearing in the last step above is independent of $x \in \Delta_n$; we show next that this quantity is, almost surely as $k \to \infty$, converging to a positive limit. Define now, in the setting of Proposition \ref{prop:extremal-Wishart-MP}, for $t>0$
$$\tilde a_t = \begin{cases}
a_t, &\qquad \text{ if } t > 1\\
a_{1/t}, &\qquad \text{ if } t < 1\\
\end{cases}$$
to be the left-most positive element of the support of the free Poisson distribution $\pi_t$. Note that we assume $t \neq 1$ and thus $\tilde a_t >0$.

Let us consider first the case $c \neq n$. Using the fact that, almost surely, as $k\to\infty$,  $s_k \lambda^\downarrow_1 \to b_{c/n}$ and $s_k \lambda^\downarrow_{s_k} \to \tilde a_{c/n} > 0$, we conclude that the first negative term in the bound \eqref{eq:bound-ARED-1st-unbalanced} vanishes, while the second term converges almost surely, as $k\to\infty$, to $\varepsilon/(nc) \tilde{a}_{c/n}$.

Let us now treat the case $c=n$; this case requires special treatment because the left edge of the support of $\pi_1$ is 0. Write, as before, for $k$ large enough,
\begin{align*}
\langle \lambda^\downarrow_{\rho_k}, \hat x^\uparrow \rangle &\geq \lambda^\downarrow_1 \eta_r + x_2(\lambda^\downarrow_{(n-2)k+2} + \lambda^\downarrow_{(n-2)k+3} + \cdots + \lambda^\downarrow_{(n-1)k}) \\
& \geq \lambda^\downarrow_1 \eta_r + x_2 \lfloor (k-1)/2 \rfloor \lambda^\downarrow_{\lceil (n-3/2)k \rceil}  \\
& \geq -\lambda^\downarrow_1  + n^{-1} \lfloor (k-1)/2 \rfloor \lambda^\downarrow_{\lceil (n-3/2)k \rceil},
\end{align*}
where $\lceil\cdot\rceil$ denotes the ceiling function and $\lfloor\cdot\rfloor$ denotes the floor function.
As before, we have that $\lambda^\downarrow_1 \to 0$, while $s_k \lambda^\downarrow_{\lceil (n-3/2)k \rceil}$ converges almost surely, as $k \to \infty$, to the $3/(2n)$-th quantile of the free Poisson distribution $\pi_1$, which is positive (see Proposition \ref{prop:extremal-Wishart-MP}). This concludes the proof.
\end{proof}

\subsection{\texorpdfstring{Threshold for $\mathrm{ARED}$ in the second unbalanced case}{Threshold for ARED in the second unbalanced case}}

\begin{theorem}\label{thm:1}
Consider the \emph{second unbalanced} asymptotic regime, where  $n \to \infty$, $k$ is fixed and $s_n\sim cnk$ for some constant $c>0$. Let $(\rho_n)$ be a sequence of random  states, where $\rho_n$ comes  from  the induced ensemble distributed according to the induced measure $\mu_{nk,s_n}$. Almost surely, when  $n\to\infty$ and $s_n\sim cnk$, one has:
\begin{enumerate}
\item If $c > \left(1+\frac{2}{k}+\frac{2}{k}\sqrt{k+1}\right)^2$, then $\rho_n\in \mathrm{ARED}_{n,k};$
\item If $c<\left(1+\frac{2}{k}+\frac{2}{k}\sqrt{k+1}\right)^2$, then $\rho_n \not\in \mathrm{ARED}_{n,k}.$
\end{enumerate}
\end{theorem}
\begin{proof} To prove the first statement we use again Lemma \ref{lem:multiple-smallest-lambda}. Indeed, by Proposition \ref{prop:extremal-Wishart-MP}, it follows that, almost surely, as $n\to\infty$,
$$s_n\left[ (k+1) \lambda^\downarrow_{nk}-\lambda^\downarrow_{1}\right]\to (k+1)a_c-b_c=k (\sqrt{c})^2-2(k+2)\sqrt{c}+k>0,$$
if $\sqrt{c}>1+\frac{2}{k}+\frac{2}{k}\sqrt{k+1}$ and thus $\rho_n\in \mathrm{ARED}_{n,k}$.

In the other case, for $x = (1/k,1/k,\ldots, 1/k) \in \Delta_{k},$ we have
$$\hat x = (\underbrace{1/k, \ldots, 1/k}_{k^2-1 \text{ times}}, \underbrace{0, \ldots, 0}_{(n-k)k \text{ times}}, 1/k-1) \in \mathbb R^{nk}.$$
Thus,
$$\langle \lambda^\downarrow_{\rho_n},\hat x^\uparrow\rangle
=-\left(1-\frac{1}{k}\right)\lambda^\downarrow_1+\frac{1}{k}\left(\lambda^\downarrow_{nk}+\cdots+
\lambda^\downarrow_{(n-k)k+2}\right).$$
Using again Proposition \ref{prop:extremal-Wishart-MP}, it follows that, almost surely, as $n\to\infty$,
$$s_n\langle \lambda^\downarrow_{\rho_n},\hat x^\uparrow\rangle\to \frac{k-1}{k} \left[(k+1)a_c-b_c\right].$$
If $c\leq 1$, then $a_c=0$ and hence the limit above is negative. On the other hand, if $c>1$, then
$(k+1)a_c-b_c=k(\sqrt{c})^2-2(k+2)\sqrt{c}+k$, which is also negative if
$c<\left(1+\frac{2}{k}+\frac{2}{k}\sqrt{k+1}\right)^2$, and the proof is complete.
\end{proof}

\begin{remark}\rm
It is of interest to notice that when  $k=2$, the threshold  for  $\mathrm{ARED}$ in the second unbalanced case is $c=7+4\sqrt{3}$ and it coincides with the one obtained in the same regime for $\mathrm{APPT}$ \cite{cny}.  This result is natural since the two criteria are equivalent when the second subsystem is a qubit  (see \cite{jlnr} and the references therein).
\end{remark}

\section{\texorpdfstring{Thresholds for $\mathrm{LS}_p$}{Thresholds for LSp}}
\label{sec:LSp}

In \cite{jlnr}, the authors introduce for each $p \in [nk]$,  the set of eigenvalue vectors for which the largest eigenvalue is less or equal than the sum of the $p$ smallest:
\begin{equation}
\label{eq:def-LP}
\mathrm{LS}_p:=\{\lambda \in \Delta_{nk} \, : \, \lambda_1^\downarrow \leq \lambda_{nk-p+1}^\downarrow + \lambda_{nk-p+2}^\downarrow + \cdots +\lambda_{nk}^\downarrow\}.
\end{equation}
It is worth to mention that the set $\mathrm{LS}_p$ is of particular interest because it sets bounds for (the more complicated set) $\mathrm{ARED}_{n,k}$.
Indeed,
accordingly to \cite{jlnr}, it holds that
for $n,k\geq2$,
\begin{equation}\label{eq.LS}
\mathrm{LS}_k \subseteq \mathrm{ARED}_{n,k} \subseteq \mathrm{LS}_{2k-1}.
\end{equation}

Note that the set $\mathrm{LS}_p$ depends only on the product $d=nk$, and not on the particular values of $n$ and $k$. We compute now the threshold for the sequence of sets $\{\mathrm{LS}_{p_d}\}_d$ in three different cases: $p_d = p$ is a fixed function, $1 \ll p_d = o(d)$ and $p_d = \lfloor td \rfloor$, for some fixed fraction $t \in (0,1)$.

\begin{theorem}\label{thm:threshold-$LS_p$}
For $d=nk$,  let  $\rho_d$ be a random induced state distributed according to the induced measure $\mu_{d,s}$.  Then, in the asymptotic regime $d\to \infty$ and $s=s_d\sim cd$ for $c>0$, one has:

\begin{itemize}
\item For every fixed integer $p\geq 2$, almost surely, as $d\to\infty$,
\begin{enumerate}
\item if $c > (1 + \frac{2}{\sqrt{p}-1})^2$, then $\rho_d \in \mathrm{LS}_{p}$;
\item if $c < (1 + \frac{2}{\sqrt{p}-1})^2$, then $\rho_d \not\in \mathrm{LS}_{p}$.
\end{enumerate}
\item For every function $p_d$ such that $1 \ll p_d = o(d)$, almost surely, as $d\to\infty$,
\begin{enumerate}
\item if $c > 1$, then $\rho \in \mathrm{LS}_{p_d}$;
\item if $c < 1$, then $\rho \not\in \mathrm{LS}_{p_d}$.
\end{enumerate}
\item For fixed $t \in (0,1)$ and $p_d = \lfloor td \rfloor$, almost surely, as $d\to\infty$,
\begin{enumerate}
\item if $c > 1-t$, then $\rho \in \mathrm{LS}_{\lfloor td \rfloor}$;
\item if $c < 1-t$, then $\rho \not\in \mathrm{LS}_{\lfloor td \rfloor}$.
\end{enumerate}
\end{itemize}
\end{theorem}
\begin{proof}
We start with the case when $p$ is fixed. For every fixed  $p\geq 2$, the inequality from \eqref{eq:def-LP} becomes asymptotically 
$b_c \leq p a_c,$
where the constants $a_c$ and  $b_c$ are defined in Proposition \ref{prop:extremal-Wishart-MP}.
In the case when $c\leq1$, $a_c=0$, so the inequality above cannot be satisfied.
In the other case, when  $c >1$, the inequality is equivalent to $c\geq (1 + \frac{2}{\sqrt{p}-1})^2$, and the conclusion of the theorem follows.

We move now to the second case where $p_d \to \infty$ as $d\to\infty$, with $p_d = o(d)$.  If $c>1$ then, almost surely, by Proposition \ref{prop:extremal-Wishart-MP},
\begin{align*}
s_d \lambda_1^\downarrow &\to b_c \\
s_d \lambda_d^\downarrow, \ldots, s_d \lambda_{d-p_d+1}^\downarrow &\to a_c >0.
\end{align*}
Using $p_d \to \infty$, we obtain $s_d (\lambda_d^\downarrow +  \cdots + \lambda_{d-p_d+1}^\downarrow) \to \infty $, finishing the proof of the first point. On the other hand, if $c<1$, the limiting measure $\pi_c$ has an atom at zero of mass $1-c$. Since the number of eigenvalues appearing on the right hand side of the inequality defining $\mathrm{LS}_{p_d}$ is $p_d = o(d)$, these eigenvalues will be zero (for $d$ large enough), while $s_d \lambda_1^\downarrow \to b_c >0$, proving the claim.

Let us now consider the final case, where the parameter $p_d$ behaves like $td$, as $d \to \infty$ ($t \in (0,1)$ is fixed). Assume first that $c > 1-t$. Notice that, since $\lambda_1^\downarrow$ converges to zero, it is enough to show that $\lambda_d^\downarrow + \cdots + \lambda_{d - \lfloor td \rfloor +1}^\downarrow$ converges to some positive limit. To show this fact, let us consider two sub-cases. First, if $c>1>1-t$, then, as $d \to \infty$,
$$\lambda_d^\downarrow + \cdots + \lambda_{d - \lfloor td \rfloor +1}^\downarrow   \geq \lfloor td \rfloor \lambda_d^\downarrow \to (t/c) a_c >0.$$
Otherwise, we have $1 \geq c > 1-t$. Take $\varepsilon>0$, small enough such that $c-3\varepsilon > 1-t$.
Then, we have, for $d$ large enough,
\begin{align*}
\lambda_d^\downarrow + \cdots + \lambda_{d - \lfloor td \rfloor +1}^\downarrow &=   \lambda_{s_d}^\downarrow + \cdots + \lambda_{d - \lfloor td \rfloor +1}^\downarrow \\
&\geq  \lambda_{\lfloor (c-\varepsilon)d \rfloor}^\downarrow + \cdots + \lambda_{d - \lfloor td \rfloor +1}^\downarrow \\
&\geq \lfloor \varepsilon d \rfloor \lambda_{\lfloor (c-\varepsilon)d \rfloor}^\downarrow.
\end{align*}
The right hand side of the above expression converges almost surely, as $d\to\infty$, to a positive constant ($\varepsilon/c$ times the $(1-c+\varepsilon)$-th quantile of $\pi_c$, which is positive), proving the first point.

We consider now the second point, when $c<1-t<1$. Since, almost surely, $\mathrm{rk} \rho_d \leq s_d \sim cd$, we have $ \lambda_d^\downarrow = \cdots = \lambda_{d - \lfloor td \rfloor +1}^\downarrow = 0$, while $s_d\lambda_1^\downarrow \to b_c >0$, showing that $\rho_d \notin \mathrm{LS}_{\lfloor td \rfloor}.$
\end{proof}

\section{\texorpdfstring{Thresholds for $\mathrm{GER}$}{Thresholds for GER}}
\label{sec:GER}

In \cite{jlnr}, the following set of probability vectors (we put $r=\min(n,k)$) was introduced:
\begin{equation}\label{eq:def-GER}
\mathrm{GER}_{n,k} = \left\{ \lambda \in \Delta_{nk} \, : \, \sum_{i=1}^{r-1} \lambda^\downarrow_i \leq 2 \lambda^\downarrow_{nk} + \sum_{i=1}^{r-1} \lambda^\downarrow_{nk-i} \right\},
\end{equation}
in connection to Hildebrand's characterization of $\mathrm{APPT}$ states. Indeed, it was show in \cite[Theorem 7.2]{jlnr} that $\mathrm{GER}_{n,k} \subseteq \mathrm{APPT}_{n,k}$; the proof consists in applying Gershgorin's circle theorem to show that Hildebrand's conditions from \cite{hil} are satisfied.

We compute next the thresholds for the set $\mathrm{GER}_{n,k}$. Note that the definition of the set is symmetric in $n$ and $k$, so we shall assume, without loss of generality, that $n \geq k$.

\begin{theorem}\label{thm:threshold-GER-balanced}
Consider the \emph{balanced} asymptotic regime, where  $n \to \infty$ and $k_n \to \infty$ as $n\to\infty$. Let $\rho_n$ be a random induced state distributed according to the induced measure $\mu_{nk_n,s_n}$. Almost surely, as $n \to\infty$ and $s_n\sim cnk_n^3$ for some constant $c>0$, one has:
\begin{enumerate}
\item[(i)]  If $c > 4$, then $\rho_n \in \mathrm{GER}_{n,k_n}$;
\item[(ii)] If $c < 4$, then $\rho_n \not\in \mathrm{GER}_{n,k_n}$.
\end{enumerate}
\end{theorem}
\begin{proof}
We consider the eigenvalues of the corresponding Wishart matrix with parameters $(nk_n, cnk_n^3)$. Note that the ratio of parameters is $ck_n^2 \to \infty$ as $n\to\infty$, so we can apply Proposition \ref{prop:Wishart-centered-normalized}. We shall consider the rescaled eigenvalues
$$\tilde \lambda_i = \lambda_i^\downarrow\left( \frac{W_n}{cnk_n^3} \right).$$
From Proposition \ref{prop:Wishart-dirac} we know that, for all $i$, $\tilde \lambda_i \to 1$ in probability, while Proposition \ref{prop:Wishart-centered-normalized} implies that, for any function $j_n = o(nk_n)$, almost surely as $n \to \infty$,
\begin{align*}
\forall \, i=1,2,\ldots, j_n , \qquad &\sqrt c k_n(\tilde \lambda_i -1) \to 2\\
\forall \, i=1,2,\ldots, j_n , \qquad &\sqrt c k_n(\tilde \lambda_{nk_n+1-i} -1) \to -2.
\end{align*}

Let us first assume $c>4$. The inequality \eqref{eq:def-GER} for the $\tilde \lambda_i$ is implied by the following equivalent inequalities (recall that $\min(n,k_n)=k_n$):
\begin{align*}
(k_n-1) \tilde \lambda_1 &\leq (k_n+1) \tilde \lambda_{nk_n} \\
k_n(\tilde \lambda_1 - \tilde \lambda_{nk_n}) &\leq \tilde\lambda_1 + \tilde\lambda_{nk_n}.
\end{align*}
The left-hand side of the inequality above converges, almost surely, as $n \to \infty$, to $4/\sqrt c$, while the right-hand side converges to $2$; since $c>4$, the conclusion follows.

In the case $c<4$, we write
\begin{align*}
\tilde \lambda_1 + \cdots  + \tilde \lambda_{k_n-1} &\geq (k_n-1) \tilde \lambda_{k_n-1} \quad \text{ and } \\
2 \tilde \lambda_{nk_n} + \tilde \lambda_{nk_n-1} + \cdots + \tilde \lambda_{nk_n-k_n+1} &\leq (k_n+1)  \tilde \lambda_{nk_n-k_n+1}.
\end{align*}
Using Proposition \ref{prop:Wishart-centered-normalized} with $j_n = k_n = o(nk_n)$, we can conclude as before.
\end{proof}

\begin{theorem}\label{thm:threshold-GER-unbalanced}
Consider the \emph{unbalanced} asymptotic regime, where $k$ is fixed and $n \to \infty$. Let $\rho_n$ be a random induced state distributed according to the induced measure $\mu_{nk,s_n}$. Almost surely, as $n \to\infty$ and $s_n\sim cnk$ for $c>0$, one has:
\begin{enumerate}
\item[(i)] If $c > (k + \sqrt{k^2-1})^2$, then $\rho_n \in \mathrm{GER}_{n,k}$;
\item[(ii)] If $c < (k + \sqrt{k^2-1})^2$, then $\rho_n \not\in \mathrm{GER}_{n,k}$.
\end{enumerate}
\end{theorem}
\begin{proof}
As usual, since the definition of the set $\mathrm{GER}$ is scale-invariant, we shall consider a Wishart matrix of parameters $(nk, cnk)$. The inequality from \eqref{eq:def-GER} reads now asymptotically (recall that $k$ is fixed)
$$(k-1)b_c \leq (k+1) a_c,$$
where the constants $a_c, b_c$ are defined in Proposition \ref{prop:extremal-Wishart-MP}. If $c\leq1$, $a_c=0$, so the inequality above cannot be satisfied; we assume thus $c > 1$. The inequality is then easily seen to be equivalent to $c \geq (k+\sqrt{k^2-1})^2$, and the proof is complete.
\end{proof}

\begin{remark}\rm
The thresholds for $\mathrm{GER}$ are the same as the thresholds for $\mathrm{APPT}$ given in \cite[Theorem 4.1]{cny} and \cite[Theorem 4.2]{cny}.
This result strengthens the claim that $\mathrm{GER}$ is a very good approximation to $\mathrm{APPT}$.
Similar results to support this claim are \cite[Remark 7.4]{jlnr} and \cite[Proposition 8.2]{jlnr}.
\end{remark}

\section{\texorpdfstring{Threshold for $\mathrm{SEPBALL}$}{Threshold for SEPBALL}}
\label{sec:SEPBALL}

One of the earlier results about the geometry of the set of separable states is a very surprising one: the largest Euclidean ball, centered at $I_d/d$, which is contained in the set of $d$-dimensional quantum states, contains only separable states \cite{gba} (to make a sense of separability, we consider an arbitrary decomposition $\mathbb C^d = \mathbb C^n \otimes \mathbb C^k$). To be more precise, the set
$$\mathrm{SEPBALL}_{n,k} =\left\{\rho\in D_{n,k} \, | \, \mathrm{Tr}(\rho^2)\leq\frac{1}{nk-1}\right\}$$
is a subset of the set of separable states. Moreover, since the set $\mathrm{SEPBALL}$ is defined in terms of the trace of the square of the density matrix, it is invariant under global unitary conjugations, so we have that $\mathrm{SEPBALL}_{n,k} \subseteq \mathrm{ASEP}_{n,k}$.

It turns out that the set $\mathrm{SEPBALL}$ is much smaller than the other sets studied in this work. In \cite[Proposition 8.2]{jlnr} it has been shown that the largest eigenvalue of elements in $\mathrm{SEPBALL}_{n,k}$ is smaller than the corresponding quantity for other sets, such as $\mathrm{APPT}_{n,k}$, $\mathrm{ARED}_{n,k}$, or $\mathrm{GER}_{n,k}$. The behavior of thresholds is also different for $\mathrm{SEPBALL}_{n,k}$ than for the other sets: the size of the ``environment'' $s_d$ scales like the \emph{square} of the total size of the system $d=nk$. In the result below, since $\mathrm{SEPBALL}_{n,k}$ depends only on the product $d=nk$, we simply write $\mathrm{SEPBALL}_d=\mathrm{SEPBALL}_{n,k}$.

\begin{theorem}\label{thm:threshold-SEPBALL}
Consider the asymptotic regime, where the total dimension $d=nk$ of the system grows (the way $n$ and $k$ grow is not relevant). Let $\rho_d$ be a random induced state distributed according to the induced measure $\mu_{d,s_d}$. Almost surely, as $d \to\infty$ and $s_d\sim cd^2$ for $c>0$, one has:
\begin{enumerate}
\item[(i)] If $c > 1$, then $\rho_d \in \mathrm{SEPBALL}_d$;
\item[(ii)] If $c < 1$, then $\rho_d \not\in \mathrm{SEPBALL}_d$.
\end{enumerate}
\end{theorem}
\begin{proof}
Replacing $\rho_d$ by $W_d / \mathrm{Tr} W_d$, where $W_d$ is a Wishart matrix of parameters $(d,s_d)$, the inequality in the definition of the set $\mathrm{SEPBALL}_d$ reads
$$(d-1)\mathrm{Tr}(W_d^2) \leq (\mathrm{Tr} W_d )^2.$$
In our setting, $d=o(s_d)$, so we can apply Proposition \ref{prop:Wishart-centered-normalized}. Replacing $W_d$ by $s_d I_d + \sqrt{ds_d}Z_d$, the above inequality simplifies to
$$ds_s (d-1) \mathrm{Tr}(Z_d^2) \leq ds_d^2 + 2 s_d \sqrt{ds_d} \mathrm{Tr} Z_d+d s_d (\mathrm{Tr} Z_d)^2.$$
Since $Z_d$ converges almost surely to a semicircular distribution, we have that $d^{-1} \mathrm{Tr} Z_d \to 0$ and  $d^{-1} \mathrm{Tr}(Z_d^2) \to 1$. Hence, the previous inequality becomes, after replacing $s_d$ by $cd^2$ and keeping only the dominating terms in $d \to \infty$,
$$cd^5 \leq c^2 d^5,$$
finishing the proof.
\end{proof}

\section{Conclusions and open questions}
\label{sec:conclusions}

In this final section we gather results about the thresholds for different entanglement criteria considered in the literature and also for some related sets in the balanced and unbalanced asymptotic regimes.

\begin{table}
\begin{center}
\begin{tabular}{cc|c|c|c|}
\cline{2-4}
&\multicolumn{1}{|c|}{\multirow{1}{*}{Balanced regime}} &\multicolumn{2}{|c|}{Unbalanced regime}\\
&\multicolumn{1}{|c|}{\multirow{1}{*}{$n,k\to\infty$}}
&\multicolumn{2}{|c|}{\multirow{1}{*}{$m=\min (n,k)$ fixed, $\max (n,k)\to\infty$}}\\
\cline{1-4}
\multicolumn{1}{ |c  }{$\mathrm{SEP}$}
&\multicolumn{1}{ |c| }{$n^3\lesssim s\lesssim n^3\log^2 n$ \cite[n=k]{asy} }
&\multicolumn{2}{ |c| }{$mnk\lesssim s\lesssim mnk\log^2 (nk)$ \cite{asy}} \\
\cline{1-4}
\multicolumn{1}{|c}{\multirow{2}{*}{$\mathrm{PPT}$ } }
&\multicolumn{1}{|c|}{$s\sim cnk$}
&\multicolumn{2}{|c|}{$s\sim cnk$}  \\
\cline{2-4}
\multicolumn{1}{ |c  }{}
&\multicolumn{1}{ |c| }{$c=4$ \cite[$n=k$]{aub}}
&\multicolumn{2}{c|}{$c=2+2\sqrt{1-\frac{1}{m^2}}$  \cite{bne}} \\
\cline{1-4}
\multicolumn{1}{ |c  }{\multirow{2}{*}{$\mathrm{RLN}$} }
&\multicolumn{1}{ |c| }{$s\sim cnk$}
&\multicolumn{2}{|c|}{$s$ fixed} \\
\cline{2-4}
\multicolumn{1}{ |c  }{}
&\multicolumn{1}{ |c| }{$c=(8/3\pi)^2$ \cite[$n=k$]{ane}}
&\multicolumn{2}{|c|}{$s=m^2$ \cite{ane}} \\
\cline{1-4}
\multicolumn{1}{|c}{\multirow{2}{*}{$\mathrm{RED}$} }
&\multicolumn{1}{ |c| }{$s\sim cn$}
&For $m=n$, $s$ is fixed &For $m=k$, $s=cnk$  \\
\cline{2-4}
\multicolumn{1}{ |c  }{}
&\multicolumn{1}{ |c| }{$c=1$ }
&$s=n$
&$c=\frac{(1+\sqrt{k+1})^2}{k(k-1)}$ \cite{jln} \\
\cline{1-4}
\end{tabular}
\end{center}
\caption {Thresholds for separability vs. entanglement and entanglement criteria} \label{tab1}
\end{table}

In Table \ref{tab1} we review the thresholds for separability ($\mathrm{SEP}$) vs. entanglement and also for  some well-known entanglement criteria: positive partial transpose ($\mathrm{PPT}$) criterion, realignment ($\mathrm{RLN}$) criterion and reduction ($\mathrm{RED}$) criterion.
In \cite{ane} the authors showed that, in the balanced case ($n=k\to\infty$), the threshold for the realignment criterion is $c=(8/3\pi)^2\approx 0.72$, on the scale $s\sim cnk$, which means that the realignment criterion is asymptotically weaker than the $\mathrm{PPT}$ criterion (from a volume perspective; from a set-inclusion perspective, the two criteria are not comparable). In the unbalanced regime it is shown that the threshold is $s=m^2$, on the scale of bounded $s$; here, $m = \min(n,k)$. On the other hand, the corresponding one for the $\mathrm{PPT}$ criterion is unbounded with respect to $\max(n,k)$, $s\sim c \max(n,k)m$. For the reduction criterion, the unbalanced case splits into two different cases depending on the parameter which tends to infinity (the dimension of the subsystem to which the reduction map is applied is important, for more comments see \cite{jlnr}). When the dimension $k$ of the second subsystem is larger, then the threshold is $s=n$ ($n$ is the dimension of the first subsystem), on the scale of bounded $s$. In the other case ($n\to\infty$ and $k$ is fixed), the threshold for the reduction criterion is $c=\frac{(1+\sqrt{k+1})^2}{k(k-1)}$, on the scale $s\sim cnk$, which is smaller than the corresponding one for the $\mathrm{PPT}$ criterion, $c=2+2\sqrt{1-\frac{1}{k^2}}$, which follows from the paper of Banica and Nechita  \cite{bne}. Moreover, for $k=2$ the two values are the same. This is natural since the reduction criterion is in general weaker than the $\mathrm{PPT}$ criterion, the two criteria being equivalent for $k=2$. In \cite{lan}, upper bounds for the threshold of $k$-extendible states have been obtained, but we do not discuss these results here, since the $k$-extendibility criterion does not enter the framework of this work.

\begin{table}
\begin{center}
\begin{tabular}{cc|c|c|c|}
\cline{2-4}
&\multicolumn{1}{|c|}{\multirow{1}{*}{Balanced regime}} &\multicolumn{2}{|c|}{Unbalanced regime}\\
&\multicolumn{1}{|c|}{\multirow{1}{*}{$n,k\to\infty$}}
&\multicolumn{2}{|c|}{\multirow{1}{*}{$m=\min (n,k)$ fixed, $\max (n,k)\to\infty$}}\\
\cline{1-4}
\multicolumn{1}{|c}{\multirow{2}{*}{$\mathrm{APPT}$ } }
&\multicolumn{1}{|c|}{$s\sim c\min(n,k)^2nk$}
&\multicolumn{2}{|c|}{$s\sim cnk$}  \\
\cline{2-4}
\multicolumn{1}{ |c  }{}
&\multicolumn{1}{ |c| }{$c=4$ \cite{cny}}
&\multicolumn{2}{c|}{$c=\left(m+\sqrt{m^2-1}\right)^2$  \cite{cny}} \\
\cline{1-4}
\multicolumn{1}{ |c  }{\multirow{2}{*}{$\mathrm{GER}$} }
&\multicolumn{1}{ |c| }{$s\sim c\min(n,k)^2nk$}
&\multicolumn{2}{|c|}{$s\sim cnk$} \\
\cline{2-4}
\multicolumn{1}{ |c  }{}
&\multicolumn{1}{ |c| }{$c=4$}
&\multicolumn{2}{|c|}{$c=\left(m+\sqrt{m^2-1}\right)^2$ } \\
\cline{1-4}
\multicolumn{1}{|c}{\multirow{2}{*}{$\mathrm{ARED}$} }
&\multicolumn{1}{ |c| }{$s\sim cnk$}
&For $m=n$, $s\sim ck$ &For $m=k$, $s=cnk$  \\
\cline{2-4}
\multicolumn{1}{ |c  }{}
&\multicolumn{1}{ |c| }{$c=1$ }
&$c=n-2$
&$c=\left(1+\frac{2}{k}+\frac{2}{k}\sqrt{k+1}\right)^2$ \\
\cline{1-4}
\end{tabular}
\end{center}
\caption {Thresholds for related sets} \label{tab2}
\end{table}

In Table \ref{tab2} we gather the thresholds for the set of absolutely PPT states ($\mathrm{APPT}$) from \cite{cny} and for $\mathrm{GER}$ and $\mathrm{ARED}$. For $k=2$, the threshold value computed for $\mathrm{ARED}$ reads $c=7+4\sqrt{3}$; this value coincides with the one obtained for $\mathrm{APPT}$, which is in agreement with the fact that $\mathrm{APPT}_{n,2}=\mathrm{ARED}_{n,2}$.
The thresholds for $\mathrm{GER}$ are the same as the thresholds for $\mathrm{APPT}$, which shows that $\mathrm{GER}$ is a very good approximation to $\mathrm{APPT}$ (see also \cite{jlnr}). Since the sets $\mathrm{SEPBALL}$ and $\mathrm{LS}_p$  depend only on the product $d=nk$, it is sufficient to consider only one asymptotic regime ($d\to\infty$), and thus we consider a separate table (Table \ref{tab3}).

Let us finish this work with a list of open questions:
\begin{enumerate}
\item Find a description of the set $\mathrm{ARLN}$ of quantum states satisfying the absolute version of the realignment criterion; following \cite{ajr}, this is a superset of $\mathrm{APPT}$. Compute the thresholds for the set $\mathrm{ARLN}$ in different asymptotic regimes.
\item Give a simple description (or tight bounds) for the set $\mathrm{ASEP}$. Using this description, compute the thresholds for the set, in different asymptotic regimes. The values of these thresholds, compared to those for $\mathrm{APPT}$, could invalidate the conjecture \cite{ajr} that $\mathrm{ASEP} = \mathrm{APPT}$.
\end{enumerate}

\begin{table}
\begin{center}
\begin{tabular}{cc|c|c|c|}
\cline{2-4}
&\multicolumn{3}{|c|}{Asymptotic regime $d=nk\to\infty$} \\
\cline{1-4}
\multicolumn{1}{|c}{\multirow{2}{*}{$\mathrm{SEPBALL}$ } }
&\multicolumn{3}{|c|}{$s\sim cd^2$} \\
\cline{2-4}
\multicolumn{1}{ |c  }{}
&\multicolumn{3}{|c|}{$c=1$} \\
\cline{1-4}
\multicolumn{1}{|c}{\multirow{3}{*}{$\mathrm{LS_{p}}$} }
&\multicolumn{3}{ |c| }{$s\sim cd$}\\
\cline{2-4}
\multicolumn{1}{ |c  }{}
&\multicolumn{1}{ |c| }{$p\geq 2$ fixed }
&$1\ll p=o(d)$
&$p=\lfloor td\rfloor$, $t\in (0,1)$ \\
\cline{2-4}
\multicolumn{1}{ |c  }{}
&\multicolumn{1}{ |c| }{$c=\left(1+\frac{2}{\sqrt{p}-1}\right)^2$ }
&$c=1$
&$c=1-t$ \\
\cline{1-4}
\end{tabular}
\end{center}
\caption {Thresholds for $\mathrm{SEPBALL}$ and $\mathrm{LS}_p$ } \label{tab3}
\end{table}


\begin{thebibliography}{99}

\bibitem{ajr}
Arunachalam, S., Johnston, N., and Russo, V.
{\it Is absolute separability determined by the partial transpose?}
 Quantum Inform. Comput. 15 (7\&8), 0694--0720 (2015).

\bibitem{aub}
Aubrun, G.
{\it Partial transposition of random states and non-centered semicircular distributions.}
Random Matrices: Theory Appl. 01, 1250001 (2012).

\bibitem{ane}
Aubrun, G. and Nechita, I.
{\it Realigning random states,}
J. Math. Phys. {53}, 102210 (2012).

\bibitem{asy}
Aubrun, G., Szarek, S.J., and Ye, D.
{\it Entanglement thresholds for random induced states.}
Comm. Pure Appl. Math. {67}, 129--171 (2014).

\bibitem{bsi}
Bai, Z. and Silverstein, J.W.
{\it Spectral Analysis of Large Dimensional Random Matrices.}
Ser.  Statist., 2010.

\bibitem{byi}
Bai, Z. D. and Yin Y. Q.
{\it Limit of the smallest eigenvalue of a large dimensional sample covariance matrix.}
Ann. Probab. 21 (3), 1275--1294 (1993).

\bibitem{ben} Bennett, C. H., Brassard, G., Popescu, S., Schumacher, B., Smolin, J. A., and Wootters, W. K.
{\it Purification of noisy entanglement and faithful teleportation via noisy channels,} Phys. Rev. Lett. 76, 722 (1996); Erratum Phys. Rev. Lett. 78, 2031 (1997).

\bibitem{bne}
Banica, T. and Nechita, I.
{\it Asymptotic eigenvalue distributions of block-transposed Wishart matrices,}
J. Theoret. Probab. {26}, 855--869 (2013).


\bibitem{cerf}
Cerf N.J., Adami C., and Gingrich R.M.
{\it Reduction criterion for separability.}
Phys. Rev. A 60, 898--909 (1999).

\bibitem{cwu}
Chen, K. and Wu, L.-A.
{\it A matrix realignment method for recognizing entanglement.}
Quantum Inform. Comput. 3 (3), 193--202 (2003).


\bibitem{cny}
Collins, B., Nechita, I., and Ye, D.
{\it The absolute positive partial transpose property for random induced states.}
Random Matrices: Theory Appl. 01, 1250002 (2012).

\bibitem{gba}
Gurvits, L. and Barnum, H.
{\it Largest separable balls around the maximally mixed bipartite quantum state.}
Phys. Rev. A 66, 062311 (2002).

\bibitem{hil}
Hildebrand, R.
{\it Positive partial transpose from spectra.}
Phys. Rev. A 76, 052325 (2007).

\bibitem{horodeckiPPTcriterion}
Horodecki, M., Horodecki, P., and  Horodecki, R.
{\it Separability of mixed states: necessary and sufficient conditions.}
Phys. Lett. A 223, 1--8 (1996).

\bibitem{horo98}
Horodecki, M., Horodecki, P., and  Horodecki, R.
{\it Mixed-state entanglement and distillation: Is there a “bound” entanglement in nature?}
Phys. Rev. Lett.  80 (24), 5239--5242 (1998).

\bibitem{horodeckireduction}
Horodecki M., Horodecki P.
{\it Reduction criterion of separability and limits for a class of distillation protocols.}
Phys. Rev. A 59, 4206--4216 (1999).

\bibitem{horodeckireview}
Horodecki R., Horodecki P., Horodecki M., and Horodecki K.
{\it Quantum entanglement.}
Rev. Mod. Phys. 81, 865--942 (2009).

\bibitem{jln}
Jivulescu, M.A., Lupa, N., and Nechita, I.
{\it On the reduction criterion for random quantum states.}
J. Math. Phys. 55, 112203 (2014).

\bibitem{jlnr}
Jivulescu, M.A., Lupa, N., Nechita, I., and Reeb, D.
{\it Positive reduction from spectra.}
Linear Algebra Appl. 469, 276-304 (2015).

\bibitem{joh}
Johnston, N.
{\it Separability from spectrum for qubit-qudit states.}
Physical Review A, 88:062330 (2013).

\bibitem{knil}
Knill, E.
{\it Separability from spectrum.}
Published electronically at http://qig.itp.unihannover.
de/qiproblems/15, 2003.

\bibitem{kus}
Ku\'{s}, M. and \.{Z}yczkowski, K.
{\it Geometry of entangled states.}
Phys. Rev. A 63, 032307 (2001).

\bibitem{lan}
Lancien, C.
{\it $k$-extendibility of high-dimensional bipartite quantum states.}
Preprint arXiv:1504.06459.

\bibitem{mpa}
Mar\v{c}enko, V.A. and Pastur, L.A.
{\it Distribution of eigenvalues for some sets of random matrices.}
Math. USSR Sb. 1, 457--483 (1967).

\bibitem{nec}
Nechita, I.
{\it Asymptotics of random density matrices.}
Ann. Henri Poincar\'{e} 8 (2007), 1521-1538.

\bibitem{peres}
Peres, A.
{\it Separability criterion for density matrices.}
Phys.
Rev. Lett. 77 (1996), 1413--1415.

\bibitem{psh}
Pastur, L. and Shcherbina, M.
{\it Eigenvalue Distribution of Large Random Matrices.}
Math. Surveys Monog., vol. 171, 2011.

\bibitem{rud}
Rudolph, O.
{\it Some properties of the computable cross-norm criterion for separability.}
Phys. Rev. A 67, 032312 (2003).

\bibitem{verstraeteaudenaert}
Verstraete, F., Audenaert, K., and De Moor, B.
{\it Maximally entangled mixed states of two qubits.}
Phys. Rev. A 54, 012316 (2001).

\bibitem{vidaltarrach}
Vidal, G., Tarrach, R.
{\it Robustness of entanglement.}
Phys. Rev. A 59, 141--155 (1999).

\bibitem{zycz}
\.{Z}yczkowski, K., Horodecki, P., Sanpera, A, and Lewenstein, M.
{\it Volume of the set of separable states.}
Phys. Rev. A 58, 883 (1998).

\bibitem{wernerentangled}
Werner, R.F.
{\it Quantum states with Einstein-Podolsky-Rosen correlations admitting a hidden-variable model.}
Phys. Rev. A 40, 4277--4281 (1989).

\end{thebibliography}
\end{document}